\documentclass[11pt, oneside]{amsart}
\usepackage{amsmath}
\usepackage[english,  activeacute]{babel}
\usepackage[latin1]{inputenc}
\usepackage{amssymb}
\usepackage{amsthm}
\usepackage{MnSymbol}
\usepackage{gastex}
\usepackage{graphics}
\usepackage{array}
\usepackage{graphicx,pstricks,pst-all,multicol}
\usepackage{array}
\usepackage{pdflscape}
\usepackage{a4wide}
\usepackage{url}

\theoremstyle{plain}
\newtheorem{theorem}{Theorem}
\newtheorem{corollary}[theorem]{Corollary}

\newtheorem{proposition}[theorem]{Proposition}
\theoremstyle{definition}
\newtheorem{definition}[theorem]{Definition}
\newtheorem{example}[theorem]{Example}
\title[Infinite Weighted Automata and Graphs]{Applications in Enumerative Combinatorics of Infinite Weighted Automata and Graphs}

\author{Rodrigo De Castro}
\address{\noindent Departamento de Matemáticas,  Universidad Nacional de Colombia, Bogotá, COLOMBIA}
\email{rdecastrok@unal.edu.co}

\author{ Andrés L. Ramírez}
\address{\noindent Departamento de Matemáticas,  Universidad Nacional de Colombia, Bogotá, COLOMBIA}
\email{andres1729@yahoo.com}

\author{ Jos\'e L. Ram\'{\i}rez}
\address{\noindent Instituto de Matem\'aticas y sus Aplicaciones, Universidad Sergio Arboleda, Bogot\'a,  COLOMBIA, \and , Departamento de Matemáticas,  Universidad Nacional de Colombia, Bogotá, COLOMBIA }
\email{josel.ramirez@ima.usergioarboleda.edu.co}
\urladdr{http://sites.google.com/site/ramirezrjl}

\date{\today}

\begin{document}
\newpsobject{grilla}{psgrid}{subgriddiv=1, griddots=10, gridlabels=6pt}

\begin{abstract}
In this paper  we present a general methodology to solve a wide variety of  classical lattice path counting problems in an uniform way. These counting problems are related to Dyck paths, Motzkin paths and some generalizations. The methodology uses weighted automata, equations of ordinary generating functions and continued fractions. It is a variation of the technique  proposed by J. Rutten, which is called Coinductive Counting.
\end{abstract}
\maketitle

\section{Introduction}

Formal languages, finite automata and grammars are basic mathematical objects in Theoretical Computer Science with remarkable applications in various fields such as  Algebra, Number Theory and Combinatorics, \cite{CAN, Mansour, DEL}. In particular, by using  finite automata and grammars,  some combinatorial results are related to enumeration of  discrete objects and their generating functions  \cite{BOU, DEL, FLA2, mansour3, mansour2}.

Combinatorics is an important branch of Mathematics which is focused on study of discrete objects. These kind of objects often arise in Theoretical Computer Science. Enumerate Combinatorics is one of the main subfield of Combinatorics and it addresses the problem how to count the number of elements of a finite set in an exact or approximate way. The finite set is given by some combinatorial conditions. Some examples of combinatorial objects are lattice paths, trees, polyominoes, words and planar maps, etc.

Several different methods exist to study combinatorial objects. For example,  the symbolic enumeration method \cite{FLA2}, the transfer-matrix method  \cite{GRAH}, the Schützenberger  methodology also called \linebreak Delest-Viennot-Schützenberger methodology \cite{BOU} and coinductive counting \cite{RUT}. The last  methodology uses infinite weighted automata, stream bisimulation and stream calculus.

An infinite weighted automaton  is a generalization of a nondeterministic weighted finite automaton. Its transitions carry weights. These weights may model, e.g., the cost involved when executing a transition, or the probability or reliability of its successful execution \cite{WEG}. Some restricted relations between infinite weighted automata and combinatorics have been developed in  \cite{ALA, RUT, UCHI, UCHI2}; however, there are few studies in this direction.

In this paper we present a variation of the methodology developed by J. Rutten in \cite{RUT}, without the use of coinduction.  We use only infinite weighted automata, weighted graphs and continued fractions. We develop a set of properties which allow us  to find the associated ordinary generating function by solving  systems of equations. It is possible to derive existing and new counting results  from this new perspective.  Specifically, we apply this methodology in problems of lattice paths, such as Dyck paths, Riordan paths,  Motzkin paths,  colored Motzkin paths, generalized Motzkin paths, among others.

The outline of this paper is as follows. In Section 2 we recall the notions of weighted automata and its ordinary generating functions. In Section 3 we develop a set of properties to decide when an automaton is convergent, which are necessary to find the ordinary generating function associated with the automaton. Then in the Section 4 we describe the proposed Counting Automata Methodology. Finally, in Section 5 we find the generating function of a family of infinite weighted automata and we work out  some examples of lattice paths.

\section{Weighted Automata and Generating Functions}
The terminology and notations are mainly those of Sakarovitch \cite{JAC} and Shallit \cite{SHA}. Let $\Sigma$ be a finite alphabet, whose elements are called \emph{symbols}. A \emph{word} over  $\Sigma$  is a finite sequence of symbols from $\Sigma$. The set of all words over $\Sigma$, i.e., the free monoid generated by $\Sigma$, is denoted by $\Sigma^*$. The identity element $\epsilon$ of $\Sigma^*$ is called the \emph{empty word}. For any word $w\in\Sigma^*$,  $\left|w\right|$  denotes its \emph{length}, i.e., the number of symbols occurring in $w$. The length of $\epsilon$ is taken to be equal to 0. If $a\in\Sigma$ and $w\in \Sigma^{*}$, then $\left|w\right|_{a}$ denotes the number of occurrences of $a$ in $w$.  Let $\Sigma$ be a finite alphabet. Then each subset of $\Sigma^*$  is called a \emph{formal language} over $\Sigma$. The number of words of length $n$ in a language  $L$ is denoted by $L^{(n)}$.\\

An  \emph{automaton} $\mathcal{M}$ is a 5-tuple $\mathcal{M}=\left(\Sigma, Q, q_{0}, F, E\right)$, where $\Sigma$  is a nonempty input alphabet, $Q$ is a nonempty set of states of $\mathcal{M}$, $q_{0}\in Q$ is the initial state of $\mathcal{M}$,  $\emptyset \neq F\subseteq Q$ is the set of final states of $\mathcal{M}$ and $E \subseteq Q \times \Sigma \times Q$ is the set of transitions of $\mathcal{M}$.  The language recognized  by an automaton   $\mathcal{M}$ is denoted by  $L(\mathcal{M})$. If $Q, \Sigma$ and $E$ are finite sets, we say that  $\mathcal{M}$ is a finite automaton \cite{JAC}.

We often describe an automaton $\mathcal{M}$ by providing a transition diagram or a labelled graph. This is a directed graph where states are represented by circles, final states by double circles, the initial state is labelled by a headless arrow entering a state, and transitions are represented by directed arrows, labelled with a symbol.  Automata in this paper do not have useless states.

\begin{example}\label{eje1}
Consider the finite automaton $\mathcal{M}=\left(\Sigma, Q, q_{0}, F, E\right)$ where $\Sigma=\left\{a, b\right\}$, $Q=\left\{q_{0}, q_{1}\right\}$, $F=\left\{q_{0}\right\}$ and $E=\{(q_0,a,q_1), (q_0,b,q_0), (q_1,a,q_0) \}$.
The transition diagram of $\mathcal{M}$ is as shown in Figure \ref{ejediagrama}. It is easy to verify that  $L(\mathcal{M})=(b \cup aa)^*$.
\begin{figure}[h]
  \begin{center}
    \unitlength=3pt
    \begin{picture}(20,10)(0,1)
    \gasset{Nw=4,Nh=4,Nmr=2,curvedepth=0}
    \thinlines
    \node[Nmarks=ir,iangle=180](A0)(0,0){\tiny{$q_{0}$}}
    \node(A1)(20,0){\tiny{$q_{1}$}}
    \drawloop[loopangle=90](A0){$b$}
    \drawedge[curvedepth=2.5](A0,A1){$a$}
    \drawedge[curvedepth=2.5](A1,A0){$a$}
    \end{picture}
  \end{center}
  \caption{Transition diagram of $\mathcal{M}$, Example  \ref{eje1}.} \label{ejediagrama}
\end{figure}
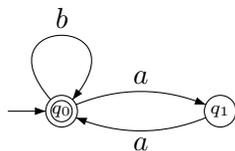

 \end{example}

\begin{example}\label{eje2}
Consider the infinite automaton $\mathcal{M_D}=\left(\Sigma, Q, q_{0}, F, E\right)$, where $\Sigma=\left\{a, b\right\}$, $Q=\left\{q_{0}, q_{1}, \ldots\right\}$, $F=\left\{q_{0}\right\}$ and $E=\left\{(q_i, a, q_{i+1}), (q_{i+1}, b, q_{i}): i\in \mathbb{N}\right\}$.
The transition diagram of $\mathcal{M_D}$ is as shown in Figure \ref{ejediagrama2}.
\begin{figure}[h]
  \centering
    \unitlength=3pt
    \begin{picture}(52, 5)(0,-1)
        \gasset{Nw=4,Nh=4,Nmr=2,curvedepth=0}
    \thinlines
    \node[Nmarks=ir,iangle=180, curvedepth=3](A0)(0,1){\tiny{$q_{0}$}}
    \node(A1)(15,1){\tiny{$q_{1}$}}
    \node(A2)(30,1){\tiny{$q_{2}$}}
    \node(A3)(45,1){\tiny{$q_{3}$}}
    \drawedge[curvedepth=2.5](A0,A1){$a$}
    \drawedge[curvedepth=2.5](A1,A0){$b$}
    \drawedge[curvedepth=2.5](A1,A2){$a$}
    \drawedge[curvedepth=2.5](A2,A1){$b$}
    \drawedge[curvedepth=2.5](A2,A3){$a$}
    \drawedge[curvedepth=2.5](A3,A2){$b$}
    \gasset{Nframe=n,Nadjust=w,Nh=6,Nmr=0}
    \node(P)(52,1){$\cdots$}
    \end{picture}
  \caption{Transition  Diagram of  $\mathcal{M_D}$, Example \ref{eje2}.}
  \label{ejediagrama2}
\end{figure}
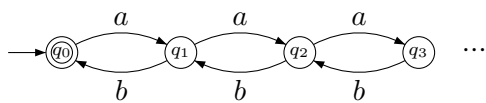

The language accepted by $\mathcal{M_D}$ is: $$L(\mathcal{M_D})=\left\{w\in\Sigma^*: |w|_a=|w|_b \ \text{and for all prefix $v$ of} \ w, |v|_b\leq|v|_a \right\}.$$  This  automaton is known as the  Dyck Automaton \cite{JEAN}.
\end{example}

\subsection{Generating Function of Languages}
An ordinary generating function  $F=\sum_{n=0}^{\infty}f_nz^n$ corresponds to a formal language $L$ if $f_n=\left|\left\{w\in L: \left|w\right|=n\right\}\right|$, i.e., if the $n$-th coefficient $f_n$ gives the number of words in $L$ with length $n$.\\ How to find the ordinary generating function (GF) corresponding to context-free language is known as the \emph{Schützenberger methodology}, (see, e.g.,   \cite{BOU,DEL,FLA2}). If $L \subseteq \Sigma^*$ is an unambiguous context-free language, then the GF corresponding to $L$ is algebraic over $\mathbb{Q}(x)$; moreover, if $L$ is a regular language then the GF corresponding to $L$  is a rational series (see, e.g.,  \cite{WEG, FLA2, JAC}).

Let $G=\left(V, \Sigma, P, S\right)$ be an unambiguous context-free grammar of the language $L_G$, where $V$ is a finite set of \emph{nonterminal symbols}, $\Sigma$ is a finite set of \emph{terminal symbols} with $\Sigma \cap V =\emptyset$,  $S\in V$ is a special symbol in $V$, called the  \emph{starting symbol} and $P$ is a finite set of \emph{rules} which,  $P\subseteq V \times \left(V \cup \Sigma\right)^*$.  The
morphism $\Theta$ is defined by
\begin{align*}
\Theta(\epsilon)&=1, \\
\Theta(a)&=z, \ \forall a \in \Sigma, \\
\Theta(A)&=A(z), \ \forall A\in V.
\end{align*}
Any production rule $A\rightarrow e_1|e_2|\cdots | e_k \in P$ yields an algebraic equation in the  $A(z), B(z), \ldots$
 \begin{align*}
 \Theta(A)=\sum_{i=1}^{\infty}\Theta(e_i).
 \end{align*}
We obtain a system of equations over the unknown  $A(z), B(z), \ldots$. This system has to be solved for $S(z)$ and gives the generating function corresponding to $L_G$.

\begin{example}
Consider the finite automaton from  Example  \ref{eje1}. Then we obtain the following system of equations
\begin{align*} \begin{cases}
\mathcal{L}_0&=\left\{b\right\} \times \mathcal{L}_0 + \left\{a\right\} \times \mathcal{L}_1 + 1\\
\mathcal{L}_1&=\left\{a\right\} \times \mathcal{L}_0
\end{cases}
\end{align*}
This gives rise to a set of equations for the associated GFs
\begin{align*} \begin{cases}
L_0(z)&=zL_0(z) + zL_1(z) + 1\\
L_1(z)&=zL_0(z)
\end{cases}
\end{align*}
Solving the system, we have the GF corresponding to $L(\mathcal{M})$.  It is
$L_0(z)$ since the initial state of the automaton is $q_0$, and
\begin{align*}
L_0(z)=\frac{1}{1-z-z^2}=\sum_{n=0}^{\infty}F_nz^n,
\end{align*}
where $F_n$ is the $n$-th Fibonacci number, see sequence A000045 \footnote{Many integer sequences and their properties are to be found electronically on the On-Line Encyclopedia of Sequences \cite{OEIS}.}.
\end{example}

\subsection{Formal Power Series and Weighted Automata}
Given an alphabet $\Sigma$ and a semiring $\mathbb{K}$. A \emph{formal power series} or \emph{formal series} $S$ is a function $S:\Sigma^*\rightarrow \mathbb{K}$. The image of a word $w$ under $S$ is called the \emph{coefficient} of $w$ in $S$ and is denoted by $s_w$. The series $S$ is written as a formal sum $S = \sum_{w\in \Sigma^*} s_ww$. The set of formal power series over $\Sigma$ with coefficients in $\mathbb{K}$ is denoted by $\mathbb{K}\left\langle\langle \Sigma^* \right\rangle\rangle$.

An automaton over $\Sigma^*$  with weight in $\mathbb{K}$, or \emph{$\mathbb{K}$-automaton} over $\Sigma^*$ is a graph labelled with elements of $\mathbb{K}\left\langle\langle \Sigma^* \right\rangle\rangle$, associated with two maps from the set of vertices to $\mathbb{K}\left\langle\langle \Sigma^* \right\rangle\rangle$. Specifically, a \emph{weighted automaton} $\mathcal{M}$ over $\Sigma^*$ with weights in $\mathbb{K}$ is 4-tuple  $\mathcal{M}=\left(Q, I, E, F\right)$ where
\begin{itemize}
  \item  $Q$ is a nonempty set of \emph{states} of $\mathcal{M}$.
  \item An element $E$ of  $\mathbb{K}\left\langle\langle \Sigma^* \right\rangle\rangle^{Q\times Q}$ called \emph{transition matrix}.
  \item $I$ is an element of $\mathbb{K}\left\langle\langle \Sigma^* \right\rangle\rangle^{Q}$, i.e., $I$ is a function from $Q$ to $\mathbb{K}\left\langle\langle \Sigma^* \right\rangle\rangle$. $I$  is the \emph{initial function} of  $\mathcal{M}$ and can also be seen as a row vector of dimension $Q$, called \emph{initial vector} of  $\mathcal{M}$.
    \item $F$ is an element of $\mathbb{K}\left\langle\langle \Sigma^* \right\rangle\rangle^{Q}$, i.e., $F$ is a function from $Q$ to $\mathbb{K}\left\langle\langle \Sigma^* \right\rangle\rangle$. $F$  is the \emph{final function} of $\mathcal{M}$ and  can also be seen as a column vector of dimension $Q$, called \emph{final vector} of  $\mathcal{M}$.
\end{itemize}
For more details see \cite{WEG, JAC}.

We say that $\mathcal{M}$ is a \emph{counting automaton} if $\mathbb{K}=\mathbb{Z}$ and $\Sigma^*=\left\{z\right\}^*$. With each automaton, we can associate a counting automaton. It can be obtained from a given automaton replacing every transition labelled with a symbol $a$, $a\in\Sigma$, by a transition labelled with $z$. This transition is called a \emph{counting transition} and the  graph is called  a \emph{counting automaton} of  $\mathcal{M}$.
\begin{figure}[h]
  \begin{center}
    \unitlength=3pt
    \begin{picture}(50, 0)(0,2)
    \gasset{Nw=5,Nh=5,Nmr=2.5,curvedepth=0}
    \thinlines
    \node[Nmarks=i,iangle=180](A0)(0,0){$p$}
    \node(A1)(20,0){$q$}
    \drawedge(A0,A1){$a$}
    \node[Nmarks=i,iangle=180](A00)(40,0){$p$}
    \node(A11)(60,0){$q$}
    \drawedge(A00,A11){$z$}
    \end{picture}
  \end{center}
  \caption{Counting transition.}\label{autconteo}
\end{figure}
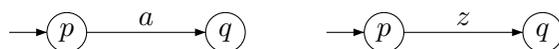

Each transition from $p$ to $q$ yields an equation:
\begin{align*}
L(p)(z)=zL(q)(z) + \left[p\in F\right] + \cdots.
\end{align*}
We use  $L_p$ to denote  $L(p)(z)$. We also use Iverson's notation, $\left[P\right] =1$ if the proposition $P$ is true and  $\left[P\right] =0$ if $P$ is false.

\section{Convergent Automata and Convergent Theorems}

We denote by $L^{(n)}(\mathcal{M})$ the number of words of length $n$ recognized by the automaton $\mathcal{M}$, including repetitions.
\begin{definition}
We say that an automaton $\mathcal{M}$ is convergent if for all integer $n\geqslant 0$, $L^{(n)}(\mathcal{M})$ is finite.
\end{definition}

It is clear that every finite automaton is convergent, however, there are non convergent infinite automata.

\begin{example}\label{eje4}
Let  $\mathcal{M}=\left(\Sigma, Q, q_{0}, F, E\right)$ be an infinite automaton, where  $\Sigma=\left\{a\right\}$,  $Q=\left\{q_{0}, q_{1},  \ldots\right\}=F$ and $E=\{(q_i,a,q_{i+1}): i\in \mathbb{N}\}$, see Figure \ref{aut2}. It is clear that  $L(\mathcal{M})=\left\{a^n: n\geq 0\right\}$, then  $L^{(n)}(\mathcal{M})=1$ for all  $n\geqslant 1$, hence  $\mathcal{M}$ is convergent.
\begin{figure}[h]
\centering
    \unitlength=3pt
    \begin{picture}(35, 3)(0,0)
    \gasset{Nw=4,Nh=4,Nmr=2,curvedepth=0}
    \thinlines
    \node[Nmarks=ir,iangle=180](A0)(0,0){\tiny{$q_{0}$}}
    \node[Nmarks=r](A1)(15,0){\tiny{$q_{1}$}}
    \node[Nmarks=r](A2)(30,0){\tiny{$q_{2}$}}
    \node[Nmarks=r](A3)(45,0){\tiny{$q_{3}$}}
    \drawedge(A0,A1){$a$}
    \drawedge(A1,A2){$a$}
    \drawedge(A2,A3){$a$}
    \gasset{Nframe=n,Nadjust=w,Nh=6,Nmr=0}
    \node(P)(51,0){$\cdots$}
    \end{picture}
  \caption{Transition diagram of $\mathcal{M}$,  Example \ref{eje4}.} \label{aut2}
\end{figure}
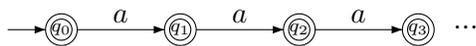
\end{example}

\begin{example}\label{eje5}
Let  $\mathcal{M}=\left(\Sigma, Q, q_{0}, F, E\right)$ be an infinite automaton, where  $\Sigma=\left\{a_{0}, a_1, \ldots\right\}$, $Q=\left\{q_{0}, q\right\}$, $F=\left\{q\right\}$ and $E=\{(q_0,a_i,q): i\in \mathbb{N}\}$, see Figure \ref{aut3}. It is clear  that $L(\mathcal{M})=\Sigma$, which is an infinite set, then $\mathcal{M}$ is not convergent.
\begin{figure}[h]
  \begin{center}
    \unitlength=3pt
    \begin{picture}(20, 5)(0,3)
    \gasset{Nw=4,Nh=4,Nmr=2,curvedepth=0}
    \thinlines
    \node[Nmarks=i,iangle=180](A0)(0,0){\tiny{$q_{0}$}}
    \node[Nmarks=r](A1)(25,0){\tiny{$q$}}
    \drawedge[curvedepth=8](A0,A1){$a_{0}$}
    \drawedge[curvedepth=4](A0,A1){$a_{1}$}
    \drawedge(A0,A1){$a_{2}$}
    \gasset{Nframe=n,Nadjust=w,Nh=6,Nmr=0}
    \node(P)(12.5,-2){$\vdots$}
    \end{picture}
  \end{center}
  \caption{Transition diagram of  $\mathcal{M}$, Example \ref{eje5}.} \label{aut3}
\end{figure}
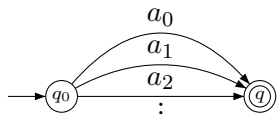
\end{example}

\begin{example}\label{eje55}
Let $\mathcal{M}=\left(\Sigma, Q, q_{0}, F, E\right)$ be an infinite automaton, where $\Sigma=\left\{a\right\}$, $Q=\left\{q_{0}, q_1,q_2,\cdots \right\}$, $F=\left\{q_1, q_2, \cdots\right\}$ and $E=\{(q_0,a,q_{i}): i\in \mathbb{Z}^+\}$, see Figure \ref{aut33}. It is clear that  $L(\mathcal{M})=\left\{a\right\}$, however, $L^{(1)}(\mathcal{M})$ is infinite.
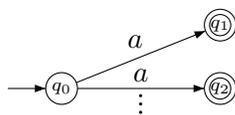
\begin{figure}[h]
  \begin{center}
    \unitlength=3pt
    \begin{picture}(20, 10)(0,3)
    \gasset{Nw=4,Nh=4,Nmr=2,curvedepth=0}
    \thinlines
    \node[Nmarks=i,iangle=180](A0)(0,0){\tiny{$q_{0}$}}
    \node[Nmarks=r](A1)(20,8){\tiny{$q_1$}}
        \node[Nmarks=r](A11)(20,0){\tiny{$q_2$}}
    \drawedge[curvedepth=0](A0,A1){$a$}
    \drawedge[curvedepth=0](A0,A11){$a$}
    \gasset{Nframe=n,Nadjust=w,Nh=6,Nmr=0}
    \node(P)(10,-2){$\vdots$}
    \end{picture}
  \end{center}
  \caption{Transition diagram of  $\mathcal{M}$, Example  \ref{eje55}.} \label{aut33}
\end{figure}
\end{example}

\subsection{Criterions for Convergence}
\begin{definition}
Let $\mathcal{M}=\left(\Sigma, Q, q_{0}, F, E\right)$ be an automaton. We defined the set of  states  of  $\mathcal{M}$ reachable   from  state  $q\in Q$ in $n$ transitions,  $Q_n^q$, recursively as follows:
\begin{align*}
Q_n^q=\begin{cases}
\left\{q\right\}, &  \ \text{if} \ n=0;\\
\bigcup\left\{p': (p,a,p')\in E,  p\in Q_{n-1}^q\right\},& \ \text{if} \  n\geqslant 1.
\end{cases}
\end{align*}
\end{definition}
\begin{theorem}[\textbf{First Convergence Theorem}]\label{teorema1conv}
Let $\mathcal{M}$ be an automaton, such that each vertex (state) of the counting automaton of  $\mathcal{M}$ has finite degree. Then $\mathcal{M}$ is convergent.
\end{theorem}
\begin{proof} Any path of length $n$ in the transition  diagram  of   $\mathcal{M}$ can be considered as a sequence of $n+1$ states, where each state is  taken exclusively of the sets
$Q_0^{q_0}, Q_1^{q_0}, Q_2^{q_0}, \ldots, Q_n^{q_0}$.  Since $\bigcup_{k=0}^{n} Q_k^{q_0}$ is a finite set,   $L^{(n)}(\mathcal{M})$ is finite, because any  word of length $n$ is obtained after $n$ choices, each with a finite number of options.
\end{proof}

\begin{example}
The counting automaton of the automaton  $\mathcal{M_D}$ in Example \ref{eje2} is convergent.
\end{example}
The following definition plays an important role in the development of applications because it allows to simplify counting automata whose transitions are formal series.

Let $\mathcal{M}$ be an automaton, and let $f(z)=\sum_{n=1}^{\infty}f_nz^n $  be a formal power series with  $f_{n} \in \mathbb{N}$ for all $n\geqslant 1$. In a counting automaton of $\mathcal{M}$ the set of counting transitions  from state $p$ to state $q$, without intermediate final states, see Figure \ref{paralelo1}(left), is represented by a graph with a single edge labelled by $f(z)$, see  Figure \ref{paralelo1}(right).

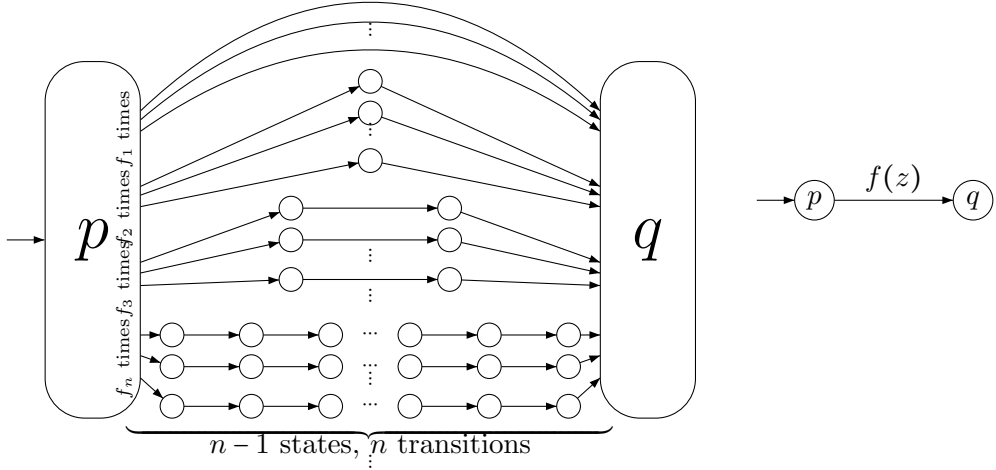
\begin{figure}[h]
\centering
   \unitlength=3pt
    \begin{picture}(106,58)(-5,0)
      \gasset{Nw=5,Nh=5,Nmr=2.5,curvedepth=0}
    \thinlines
    \node[Nmarks=i, iangle=180, Nw=12, Nadjustdist=5, Nh=45, Nmr=5](A0)(-5,30){\Huge{$p$}}
    \node[Nw=12, Nadjustdist=5, Nh=45, Nmr=5](A1)(65,30){\Huge $q$}
    \drawedge[curvedepth=20, syo=10,eyo=10](A0,A1){}
    \drawedge[curvedepth=18, syo=9.5,eyo=9.5](A0,A1){}
    \drawedge[curvedepth=15, syo=9,eyo=9](A0,A1){}
    \node[Nh=3,Nw=3,Nmr=1.5](A2)(30,50){}
     \node[Nh=3,Nw=3,Nmr=1.5](A3)(30,46){}
     \node[Nh=3,Nw=3,Nmr=1.5](A4)(30,40){}
    \drawedge[syo=4](A0,A2){}
    \drawedge[syo=3.5](A0,A3){}
    \drawedge[syo=3](A0,A4){}
    \drawedge[eyo=4](A2,A1){}
    \drawedge[eyo=3.5](A3,A1){}
    \drawedge[eyo=3](A4,A1){}
      \node[Nh=3,Nw=3,Nmr=1.5](A5)(20,34){}
     \node[Nh=3,Nw=3,Nmr=1.5](A55)(40,34){}
      \node[Nh=3,Nw=3,Nmr=1.5](A6)(20,30){}
     \node[Nh=3,Nw=3,Nmr=1.5](A66)(40,30){}
      \node[Nh=3,Nw=3,Nmr=1.5](A7)(20,25){}
     \node[Nh=3,Nw=3,Nmr=1.5](A77)(40,25){}
     \drawedge[syo=-5](A0,A5){}
     \drawedge(A5,A55){}
      \drawedge[eyo=-5](A55,A1){}
    \drawedge[syo=-5.5](A0,A6){}
     \drawedge(A6,A66){}
      \drawedge[eyo=-5.5](A66,A1){}
    \drawedge[syo=-6](A0,A7){}
     \drawedge(A7,A77){}
      \drawedge[eyo=-6](A77,A1){}
       \node[Nh=3,Nw=3,Nmr=1.5](A8)(5,18){}
     \node[Nh=3,Nw=3,Nmr=1.5](A88)(15,18){}
      \node[Nh=3,Nw=3,Nmr=1.5](A888)(25,18){}
     \node[Nh=3,Nw=3,Nmr=1.5](A82)(35,18){}
      \node[Nh=3,Nw=3,Nmr=1.5](A822)(45,18){}
     \node[Nh=3,Nw=3,Nmr=1.5](A8222)(55,18){}
        \node[Nh=3,Nw=3,Nmr=1.5](A9)(5,14){}
     \node[Nh=3,Nw=3,Nmr=1.5](A99)(15,14){}
      \node[Nh=3,Nw=3,Nmr=1.5](A999)(25,14){}
     \node[Nh=3,Nw=3,Nmr=1.5](A92)(35,14){}
      \node[Nh=3,Nw=3,Nmr=1.5](A922)(45,14){}
     \node[Nh=3,Nw=3,Nmr=1.5](A9222)(55,14){}
      \node[Nh=3,Nw=3,Nmr=1.5](A10)(5,9){}
     \node[Nh=3,Nw=3,Nmr=1.5](A110)(15,9){}
      \node[Nh=3,Nw=3,Nmr=1.5](A1110)(25,9){}
     \node[Nh=3,Nw=3,Nmr=1.5](A12)(35,9){}
      \node[Nh=3,Nw=3,Nmr=1.5](A112)(45,9){}
     \node[Nh=3,Nw=3,Nmr=1.5](A1112)(55,9){}
     \drawedge[syo=-12](A0,A8){}
     \drawedge(A8,A88){}
      \drawedge(A88,A888){}
    \drawedge(A82,A822){}
     \drawedge(A822,A8222){}
      \drawedge[eyo=-12](A8222,A1){}
    \drawedge[syo=-12.5](A0,A9){}
     \drawedge(A9,A99){}
      \drawedge(A99,A999){}
    \drawedge(A92,A922){}
     \drawedge(A922,A9222){}
      \drawedge[eyo=-12.5](A9222,A1){}
    \drawedge[syo=-12](A0,A10){}
     \drawedge(A10,A110){}
      \drawedge(A110,A1110){}
    \drawedge(A12,A112){}
     \drawedge(A112,A1112){}
      \drawedge[eyo=-12](A1112,A1){}
      \gasset{Nw=5,Nh=5,Nmr=2.5,curvedepth=0}
    \thinlines
    \node[Nmarks=i,iangle=180](A0)(86,35){$p$}
    \node(A1)(106,35){$q$}
    \drawedge(A0,A1){$f(z)$}

     \gasset{Nframe=n,Nadjust=w,Nh=6,Nmr=0}
    \node(P)(30,56.5){\footnotesize{\tiny{$\vdots$}}}
    \node(P)(30,44){\footnotesize{\tiny{$\vdots$}}}
    \node(P)(30,28){\footnotesize{\tiny{$\vdots$}}}
   \node(P)(30,12.5){\tiny{$\vdots$}}
    \node(P)(30,23){\tiny{$\vdots$}}
   \node(P)(30,18){\tiny{$\cdots$}}
   \node(P)(30,14){\tiny{$\cdots$}}
   \node(P)(30,9){\tiny{$\cdots$}}
   \node(P)(-1,44){\rotatebox{90}{\tiny $f_1$  times}}
   \node(P)(-1,34){\rotatebox{90}{\tiny$f_2$ times}}
   \node(P)(-1,25){\rotatebox{90}{\tiny$f_3$ times}}
   \node(P)(-1,15){\rotatebox{90}{\tiny $f_n$ times}}
   \node(P)(30,6){$\underbrace{ \hspace{6.5cm}  }$}
   \node(P)(30,4){$n-1$ states, $n$ transitions}
   \node(P)(30,2){\tiny{$\vdots$}}
         \end{picture}
      \caption{Transitions  from the state $p$ to $q$ and transition in parallel.}
         \label{paralelo1}
         \end{figure}
This kind of transition is called a \emph{transition in parallel}. The states $p$ and $q$  are called  \emph{visible states} and the intermediate states are called  \emph{hidden states}.

\begin{example}\label{ejeconteoo}
In Figure  \ref{ejeautconteo}(left)  we display a counting  automaton $\mathcal{M}_1$ without  transitions in parallel, i.e., every transition is label by $z$. The transitions from state  $q_1$ to  $q_2$ correspond to the series  $\frac{1-\sqrt{1-4z}}{2}=z+z^2+2z^3+5z^4+14z^5+\cdots$. However, this automaton  can also be represented using transitions in parallel. Figure  \ref{ejeautconteo}(right) displays two examples.
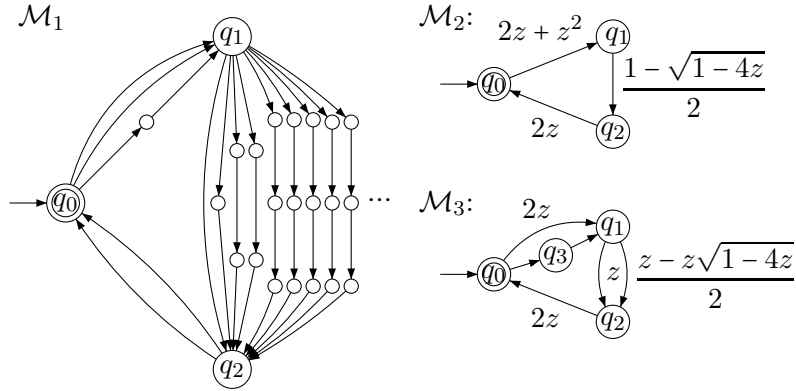
\begin{figure}[h]
\centering
   \unitlength=1.8pt
    \begin{picture}(130,78)(0,0)
      \gasset{Nw=5,Nh=5,Nmr=2.5,curvedepth=0}
    \thinlines
    \node[Nmarks=ri, iangle=180, Nw=8, Nadjustdist=8, Nh=8, Nmr=8](A0)(0,35){$q_0$}
    \node[Nw=8, Nadjustdist=8, Nh=8, Nmr=8](A1)(35,70){$q_1$}
   \node[Nw=8, Nadjustdist=8, Nh=8, Nmr=8](A2)(35,0){$q_2$}
   \node[Nw=3, Nadjustdist=3, Nh=3, Nmr=3](A011)(17,52){}
   \node[Nw=3, Nadjustdist=3, Nh=3, Nmr=3](A121)(32,35){}
   \node[Nw=3, Nadjustdist=3, Nh=3, Nmr=3](A122)(36,23){}
   \node[Nw=3, Nadjustdist=3, Nh=3, Nmr=3](A123)(36,46){}
   \node[Nw=3, Nadjustdist=3, Nh=3, Nmr=3](A1233)(40,23){}
   \node[Nw=3, Nadjustdist=3, Nh=3, Nmr=3](A1234)(40,46){}
   \node[Nw=3, Nadjustdist=3, Nh=3, Nmr=3](A124)(44,17.5){}
   \node[Nw=3, Nadjustdist=3, Nh=3, Nmr=3](A125)(44,35){}
   \node[Nw=3, Nadjustdist=3, Nh=3, Nmr=3](A126)(44,52.5){}
   \node[Nw=3, Nadjustdist=3, Nh=3, Nmr=3](A127)(48,17.5){}
   \node[Nw=3, Nadjustdist=3, Nh=3, Nmr=3](A128)(48,35){}
   \node[Nw=3, Nadjustdist=3, Nh=3, Nmr=3](A129)(48,52.5){}
   \node[Nw=3, Nadjustdist=3, Nh=3, Nmr=3](A1210)(52,17.5){}
   \node[Nw=3, Nadjustdist=3, Nh=3, Nmr=3](A1211)(52,35){}
   \node[Nw=3, Nadjustdist=3, Nh=3, Nmr=3](A1212)(52,52.5){}
   \node[Nw=3, Nadjustdist=3, Nh=3, Nmr=3](A1213)(56,17.5){}
   \node[Nw=3, Nadjustdist=3, Nh=3, Nmr=3](A1214)(56,35){}
   \node[Nw=3, Nadjustdist=3, Nh=3, Nmr=3](A1215)(56,52){}
   \node[Nw=3, Nadjustdist=3, Nh=3, Nmr=3](A1216)(60,17.5){}
      \node[Nw=3, Nadjustdist=3, Nh=3, Nmr=3](A1217)(60,35){}
         \node[Nw=3, Nadjustdist=3, Nh=3, Nmr=3](A1218)(60,52){}
   \drawedge[curvedepth=5](A0,A1){}
   \drawedge[curvedepth=8](A0,A1){}
   \drawedge(A0,A011){}
   \drawedge(A011,A1){}
   \drawedge[curvedepth=-6](A1,A2){}
   \drawedge[curvedepth=-3](A2,A0){}
  \drawedge[curvedepth=3](A2,A0){}
  \drawedge(A1,A121){}
  \drawedge(A121,A2){}
  \drawedge(A1,A1234){}
  \drawedge(A1234,A1233){}
  \drawedge(A1233,A2){}
    \drawedge(A1,A123){}
  \drawedge(A123,A122){}
  \drawedge(A122,A2){}
  \drawedge[curvedepth=1](A1,A1215){}
  \drawedge(A1215,A1214){}
  \drawedge(A1214,A1213){}
  \drawedge[curvedepth=1](A1213,A2){}
  \drawedge[curvedepth=1](A1,A1212){}
  \drawedge(A1212,A1211){}
  \drawedge(A1211,A1210){}
  \drawedge[curvedepth=1](A1210,A2){}
  \drawedge[curvedepth=1](A1,A129){}
  \drawedge(A129,A128){}
  \drawedge(A128,A127){}
  \drawedge[curvedepth=1](A127,A2){}
  \drawedge[curvedepth=1](A1,A126){}
  \drawedge(A126,A125){}
  \drawedge(A125,A124){}
  \drawedge[curvedepth=1](A124,A2){}
    \drawedge[curvedepth=1](A1,A1218){}
  \drawedge(A1218,A1217){}
  \drawedge(A1217,A1216){}
  \drawedge[curvedepth=1](A1216,A2){}
         \gasset{Nw=7,Nh=7,Nmr=3.5,curvedepth=0}
    \thinlines
    \node[Nmarks=ri, iangle=180](A0)(90,60){$q_0$}
    \node(A1)(115,70){ $q_1$}
   \node(A2)(115,50){$q_2$}
   \node[Nmarks=ri, iangle=180](AA0)(90,20){$q_0$}
    \node(AA1)(115,30){$q_1$}
   \node(AA2)(115,10){$q_2$}
   \node(AA3)(103,24){$q_3$}
   \drawedge(A0,A1){$2z+z^2$}
   \drawedge(A1,A2){$\dfrac{1-\sqrt{1-4z}}{2}$}
   \drawedge(A2,A0){$2z$}
      \drawedge[curvedepth=5](AA0,AA1){$2z$}
      \drawedge(AA0,AA3){}
      \drawedge(AA3,AA1){}
          \drawedge(AA2,AA0){$2z$}
     \drawedge[curvedepth=-3](AA1,AA2){$z$}
     \drawedge[curvedepth=3](AA1,AA2){$\dfrac{z-z\sqrt{1-4z}}{2}$}
     \gasset{Nframe=n,Nadjust=w,Nh=6,Nmr=0}
    \node(P)(-5,74){$\mathcal{M}_1$}
   \node(P)(66,35){$\cdots$}
    \node(P)(80,74){$\mathcal{M}_2$:}
        \node(P)(80,35){$\mathcal{M}_3$:}
         \end{picture}
  \caption{Counting automata with transitions in parallel, Example \ref{ejeconteoo}.}
  \label{ejeautconteo}
\end{figure}
\end{example}

This example shows that a counting automaton can have different equivalent representations.
\begin{definition}\label{autoequiv}
Two counting automata $\mathcal{M}_1$ and $\mathcal{M}_2$ are equivalent if for all integer $n\geqslant 0$,  $L^{(n)}(\mathcal{M}_1)=L^{(n)}(\mathcal{M}_2)$. This is denoted by  $\mathcal{M}_1\cong\mathcal{M}_2$.
\end{definition}

\begin{definition}
Let  $f(z)=\sum_{t=0}^{\infty}f_tz^t$ be a power series (or a polynomial). We define  $_nf(z)$ as the polynomial $_nf(z)=\sum_{t=0}^{n}f_tz^t$.
\end{definition}

\begin{theorem}[\textbf{Second Convergence Theorem}]\label{teorema2conv}
Let  $\mathcal{M}$ be an automaton, and  let  $f_1^q(z), f_2^q(z), \ldots, $ be the transitions in parallel from  state $q\in Q$ in a counting automaton of $\mathcal{M}$. Then $\mathcal{M}$ is convergent if the series
\begin{align*}
F^q(z)=\sum_{k=1}^{\infty}f_k^q(z)
\end{align*}
is a convergent series for each visible state $q\in Q$ of the counting automaton.
\end{theorem}
\begin{proof}
Any word of length $n$ is accepted by $\mathcal{M}$  if there is a path from $q_0$ to some final state.  Since the hidden states in a transition in parallel are not final states, then the paths from a visible state to another visible state, corresponding to the terms $f_{n+1}z^{n+1}, f_{n+2}z^{n+2},\ldots$ of a transition in parallel $f(z)$ in a counting automaton of $\mathcal{M}$, are not accepting paths.  Let $\mathcal{M}'$ be an automaton  obtained by replacing all transitions in parallel $f(z)$ in a counting automaton  of $\mathcal{M}$ by the transition $_nf(z)$, then clearly $L^{(n)}(\mathcal{M})=L^{(n)}(\mathcal{M}')$ for all $n\geqslant 0$. On the other hand,  the number of transitions of each visible state $q\in Q'$  in  $\mathcal{M}'$ is finite because
\begin{align*}
\sum_{k=1}^{\infty}\ _nf^q_k(1)=\ _nF^q(1)<\infty,
\end{align*}
 and from each of the hidden state starts a single transition. Hence, by Theorem \ref{teorema1conv}, $\mathcal{M}'$ is convergent. Hence $L^{(n)}(\mathcal{M})=L^{(n)}(\mathcal{M}')$ is finite for all $n\geqslant 0$, therefore  $\mathcal{M}$  is convergent.
\end{proof}

\begin{proposition}\label{sisecuacion}
If $f(z)$ is a polynomial transition in parallel from state $p$  to $q$ in a finite counting automaton  $\mathcal{M}$, then this gives rise to an equation in the  system of GFs equations of $\mathcal{M}$:
\begin{align*}
L_p=f(z)L_q + \left[p\in F\right]+ \cdots
\end{align*}
\end{proposition}
\begin{proof}
Let  $f(z)=f_1z + f_2z^2 + \cdots + f_nz^n$ be a polynomial transition in parallel, then the set of transitions in parallel  corresponding to the term  $f_kz^k$, $1\leq k \leq n$, can be
represented with a graph as in Figure \ref{tranpolino}.

\begin{figure}[h]
\centering
   \unitlength=2.5pt
    \begin{picture}(60,40)(0,0)
      \gasset{Nw=5,Nh=5,Nmr=2.5,curvedepth=0}
    \thinlines
    \node[Nmarks=i, iangle=180, Nw=10, Nadjustdist=5, Nh=10, Nmr=5](A0)(0,20){\Large{$p$}}
    \node[Nw=10, Nadjustdist=5, Nh=10, Nmr=5](A1)(60,20){\Large $q$}
    \node[Nh=9,Nw=9,Nmr=4.5](A11)(15,35){\scriptsize$q_{11}$}
    \node[Nh=9,Nw=9,Nmr=4.5](A12)(30,35){\scriptsize$q_{12}$}
    \node[Nh=9,Nw=9,Nmr=4.5](A13)(45,35){\scriptsize$q_{1k-1}$}
     \node[Nh=9,Nw=9,Nmr=4.5](A21)(15,22){\scriptsize$q_{21}$}
     \node[Nh=9,Nw=9,Nmr=4.5](A22)(30,22){\scriptsize$q_{22}$}
    \node[Nh=9,Nw=9,Nmr=4.5](A23)(45,22){\scriptsize$q_{2k-1}$}
     \node[Nh=9,Nw=9,Nmr=4.5](A31)(15,5){\tiny$q_{f_k 1}$}
     \node[Nh=9,Nw=9,Nmr=4.5](A32)(30,5){\tiny$q_{f_k 2}$}
     \node[Nh=9,Nw=9,Nmr=4.5](A33)(45,5){\tiny$q_{f_{k}k-1}$}
    \drawedge(A0,A11){$z$}
    \drawedge(A11,A12){$z$}
    \drawedge(A13,A1){$z$}
    \drawedge(A0,A21){$z$}
    \drawedge(A21,A22){$z$}
    \drawedge(A23,A1){$z$}
    \drawedge(A0,A31){$z$}
    \drawedge(A31,A32){$z$}
    \drawedge(A33,A1){$z$}
     \gasset{Nframe=n,Nadjust=w,Nh=6,Nmr=0}
    \node(P)(37.5,35){$\cdots$}
    \node(P)(37.5,22){$\cdots$}
    \node(P)(37.5,5){$\cdots$}
   \node(P)(15,14){$\vdots$}
    \node(P)(30,14){$\vdots$}
    \node(P)(45,14){$\vdots$}
          \end{picture}
          \caption{Transition in parallel corresponding to the term $f_kz^k$.}
          \label{tranpolino}
  \end{figure}
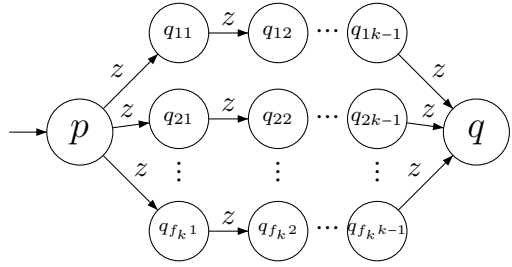

Therefore the GF equation is
\begin{align*}
L_p(z)=zL_{q_{11}}(z) + zL_{q_{21}}(z) + \cdots + zL_{q_{f_k1}}(z) + \left[p\in F\right].
\end{align*}
Since $L_{q_{i1}}(z)=z^{k-1}L_q(z)$, then
\begin{align*}
L_p(z)&=\overbrace{z^kL_{q}(z) + z^kL_{q}(z) + \cdots + z^kL_{q}(z)}^{f_k - \text{times}} + \left[p\in F\right]\\
&=f_kz^kL_q(z) + \left[p\in F\right].
\end{align*}
Considering each of the terms $f_k$, ($1\leq k \leq n$), then
\begin{align*}
L_p(z)&=f_1zL_{q}(z) + f_2z^2L_{q}(z) + \cdots + f_nz^nL_{q}(z) + \left[p\in F\right]\\
&=f(z)L_q(z) + \left[p\in F\right].
\end{align*}
\end{proof}

\begin{proposition}\label{sisecuacion2}
Let $\mathcal{M}$ be a convergent automaton such that  a counting automaton of $\mathcal{M}$ has a finite number  of visible states $q_0,q_1,\ldots,q_r$, in which the number of transitions in parallel  starting from each state is finite.   Let \linebreak
$f_1^{q_t}(z), f_2^{q_t}(z), \ldots, f_{s(t)}^{q_t}(z)$ be the transitions in parallel from the state $q_t\in Q$. Then the  GF for the  language $L(\mathcal{M})$ is $L_{q_0}(z)$. It is obtained  by solving the system of $r+1$ GFs equations
\begin{multline*}
L({q_t})(z)=f_1^{q_t}(z)L(q_{t_1})(z)+f_2^{q_t}(z)L(q_{t_2})(z)+ \cdots
+f_{s(t)}^{q_t}(z)L(q_{t_{s(t)}})(z)+[q_t\in F],
\end{multline*}
with  $0\leq t\leq r$, where $q_{t_k}$ is the visible state joined with $q_t$ through the transition in parallel $f_k^{q_t}$, and  $L(q_{t_k})$ is the GF for the language accepted by
$\mathcal{M}$ if  $q_{t_k}$ is the initial state.
\end{proposition}
\begin{proof}
Let $n \geqslant 0$ be an integer, and let  $\mathcal{M}'$ be the automaton obtained by replacing the $s(t)$ transitions in parallel ($f_1^{q_t}(z), f_2^{q_t}(z), \ldots, f_{s(t)}^{q_t}(z)$)  leaving the state $q_t\in Q$ by the transitions $_nf_1^{q_t}(z), \ _nf_2^{q_t}(z),\ldots, \ _nf_{s(t)}^{q_t}(z)$, with $0\leq t \leq r$. Hence $\mathcal{M}'$ is a finite automaton and from Proposition \ref{sisecuacion} the GF of $\mathcal{M}'$ is obtained by  solving the following system for $L'(q_0)$
\begin{multline*}
L'({q_t})(z)= \ _nf_1^{q_t}(z)L'(q_{t_1})(z)+ \ _nf_2^{q_t}(z)L'(q_{t_2})(z)+\cdots
\\ + \ _nf_{s(t)}^{q_t}(z)L'(q_{t_{s(t)}})(z)+[q_t\in F],
\end{multline*}
where $0\leq t \leq r$. Therefore $L'(q_0)(z)$ is a rational function $R$ evaluated at variables  $_nf_k^{q_t}(z)$, with $0\leq t\leq r, 1\leq k \leq s(t)$. We denoted this by $L'(q_0)(z)=R(_nf_k^{q_t}(z))$. It is clear that $_nL'(q_0)(z)=\ _nL(q_0)(z)$.  Finally we consider the series $R(f_k^{q_t}(z))$, as  the calculation of a rational expression involves only a finite number of sums, differences, products and reciprocals, and after applying one of these operations the $n$-th term of the power series depends only  on the first $n$ terms of the series involved in the operation. Hence, $R(_nf_k^{q_t}(z))=\ _nR(f_k^{q_t}(z))$ for all $n\geq0$. Therefore the GF of $\mathcal{M}$ is $R(f_k^{q_t}(z))=L(q_0)$.
\end{proof}

\begin{example}
The system of GFs equations  associated with $\mathcal{M}_2$, see Example \ref{ejeconteoo},  is
\begin{align*}
\begin{cases}
L_0&=(2z+z^2)L_1 + 1 \\
L_1&=\dfrac{1- \sqrt{1-4z}}{2}L_2\\
L_2&=2zL_0.
\end{cases}
\end{align*}
Solving the system for  $L_0$, we find the GF for the language  $\mathcal{M}_2$ and therefore of  $\mathcal{M}_1$ and  $\mathcal{M}_3$.
\begin{align*}
L_0=\frac{1}{1-(2z^2+z^3)(1-\sqrt{1-4z})}= 1 + 4z^3 + 6z^4 + 10z^5 + 40z^6 + 114z^7+ \cdots
\end{align*}
\end{example}

\section{Counting Automata  Methodology (CAM)}

A counting automaton associated with an automaton $\mathcal{M}$ can be used to model combinatorial objects if there is a bijection between all words recognized by the automaton $\mathcal{M}$  and the combinatorial objects. Such method, along with the previous theorems and propositions constitute the  \textbf{Counting Automata  Methodology (CAM)}.

We distinguish three phases in the CAM:
\begin{enumerate}
\item Given a problem of enumerative combinatorics, we have to find a convergent automaton (see Theorems \ref{teorema1conv} and \ref{teorema2conv}) whose GF is the solution of the problem.

\item Find a general formula for  the GF of $\mathcal{M}'$  ($f_i={L_0}_i$), see Propositions \ref{sisecuacion} and \ref{sisecuacion2}. The GF of the automaton $\mathcal{M}'$ is obtained from $\mathcal{M}$ by removing a set of states or edges. Sometimes we find a relation of iterative type, such as a continued fraction.  Moreover, it is practical to use equivalent automata, see Definition \ref{autoequiv}.

\item Find the GF $f(z)$ to which the GFs associated with each $\mathcal{M}'$ converge, which is  guaranteed by the convergence theorems. \end{enumerate}

\subsection{Examples of the Counting Automata Methodology}

\begin{example}\label{ejmotzkin}
A Motzkin path of length $n$ is a lattice path of  $\mathbb{Z \times Z}$ running from $(0, 0)$ to $(n, 0)$ that never passes below the $x$-axis and whose permitted steps are the up diagonal step $U=(1, 1)$, the down diagonal step $D=(1,-1)$ and the horizontal step $H=(1, 0)$, called rise, fall and level step, respectively.  The number of Motzkin paths of length $n$ is the $n$-th Motzkin number $m_{n}$, sequence  A001006.  Many other examples of  bijections between Motzkin numbers and others combinatorial objects  can be found in \cite{FRA}.

The number of words of length $n$  recognized by the convergent automaton  $\mathcal{M}_{\mathrm{Mot}}$, see Figure \ref{figmotzkin2},  is  the $n$-th Motzkin number and its GF is
\begin{align}
M\left(z\right)=\sum_{i=0}^{\infty}m_{i}z^{i}&=\frac{1-z-\sqrt{1-2z-3z^2}}{2z^2}\\
&=\cfrac{1}{1-z-\cfrac{z^2}{1-z-\cfrac{z^2}{1-z-\cfrac{z^2}{\ddots}}}}.
\end{align}
\begin{figure}[h]
  \centering
    \unitlength=2.8pt
    \begin{picture}(52, 10)(0,-1)
        \gasset{Nw=4,Nh=4,Nmr=2,curvedepth=0}
    \thinlines
    \node[Nmarks=ir,iangle=180, curvedepth=3](A0)(0,1){\tiny{$q_{0}$}}
    \node(A1)(15,1){\tiny{$q_{1}$}}
    \node(A2)(30,1){\tiny{$q_{2}$}}
    \node(A3)(45,1){\tiny{$q_{3}$}}
    \drawedge[curvedepth=2.5](A0,A1){$z$}
    \drawedge[curvedepth=2.5](A1,A0){$z$}
    \drawedge[curvedepth=2.5](A1,A2){$z$}
    \drawedge[curvedepth=2.5](A2,A1){$z$}
    \drawedge[curvedepth=2.5](A2,A3){$z$}
    \drawedge[curvedepth=2.5](A3,A2){$z$}
    \drawloop[loopdiam=4,loopangle=90](A0){$z$}
    \drawloop[loopdiam=4,loopangle=90](A1){$z$}
    \drawloop[loopdiam=4,loopangle=90](A2){$z$}
    \drawloop[loopdiam=4,loopangle=90](A3){$z$}
    \gasset{Nframe=n,Nadjust=w,Nh=6,Nmr=0}
    \node(P)(52,1){$\cdots$}
   \node(P)(-12,8){$\mathcal{M}_{\mathrm{Mot}}:$}
    \end{picture}
  \caption{Convergent Automaton associated with Motzkin Paths.} \label{figmotzkin2}
\end{figure}
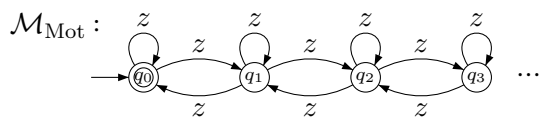

In this case the edge from state $q_{i}$ to state  $q_{i+1}$ represents a rise, the edge from the state $q_{i+1}$ to $q_{i}$ represents a fall and the loops represent the level steps, see   Table \ref{tab:mot}.
\begin{table}[h]
  \centering
  \begin{tabular}{|c|c|c|}\hline
  $(q_i, z,q_{i+1})\in E \Leftrightarrow$
\psset{unit=5mm}
\begin{pspicture}(0,0)(2,2)
\psgrid[gridwidth=0.3pt,gridcolor=gray,subgriddiv=0, gridlabels=0]
\psline[linewidth=1pt]{-}(0,0)(2,2)
\end{pspicture}
&
$(q_{i+1}, z, q_{i})\in E \Leftrightarrow$
\psset{unit=5mm}
\begin{pspicture}(0,0)(2,2)
\psgrid[gridwidth=0.3pt,gridcolor=gray,subgriddiv=0, gridlabels=0]
\psline[linewidth=1pt]{-}(0,2)(2,0)
\end{pspicture}
&
$(q_{i}, z,q_{i})\in E \Leftrightarrow$
\psset{unit=5mm}
\begin{pspicture}(0,0)(2,2)
\psgrid[gridwidth=0.3pt,gridcolor=gray,subgriddiv=0, gridlabels=0]
\psline[linewidth=1pt]{-}(0,1)(2,1)
\end{pspicture} \\ \hline
       \end{tabular}
  \caption{Bijection between $\mathcal{M}_{\mathrm{Mot}}$ and Motzkin paths.}
  \label{tab:mot}
\end{table}

Moreover, it is clear that a word is recognized by  $\mathcal{M}_{\mathrm{Mot}}$  if and only if the number of steps to the right and to the left coincide, which ensures that the path is well formed. Then $m_n=\left|\left\{w\in L(\mathcal{M}_{\mathrm{Mot}}): \left|w\right|=n \right\}\right|=L^{(n)}(\mathcal{M}_{\mathrm{Mot}})$.

Let $\mathcal{M}_{\mathrm{Mot}s}$, $s\geq 1$ be the automaton  obtained from $\mathcal{M}_{\mathrm{Mot}}$, by deleting the states $q_{s+1}, q_{s+2}, \dots$.  Therefore the system of GFs equations of $\mathcal{M}_{\mathrm{Mot}s}$ is
\begin{align*}
\left\{\begin{aligned}
L_{0}&= zL_{0}+zL_{1}+1,\\
L_{i}&= zL_{i-1} + zL_{i}+ zL_{i+1}, \ \ 1\leq i \leq s-1, \\
L_{s}&= zL_{s-1}+zL_{s}.\end{aligned}
\right.
\end{align*}
Substituting repeatedly into each equation $L_i$, we have
\begin{center}
\begin{tabular}{c}
$L_0=\cfrac{H}{ 1-\cfrac{F^2}{1-\cfrac{F^2}{\cfrac{\vdots }{1-F^2}}}}$
\end{tabular}
\hspace{-3ex}
\raisebox{-2ex}{$\left.\phantom{\begin{matrix} 1 \\ 1 \\ 1  \\ 1 \\ 1
\end{matrix}} \right\}$}
\  \raisebox{-2ex}{$s$ times}
\end{center}
where $F=\frac{z}{1-z}$ and $H=\frac{1}{1-z}$.  Since $\mathcal{M}_{\mathrm{Mot}}$ is convergent, then  as $s\rightarrow \infty$ we obtain a convergent continued fraction    $M$ of the  GF of $\mathcal{M}_{\mathrm{Mot}}$. Moreover,
\begin{align*}
M=\cfrac{H}{1-F^2\left(\frac{M}{H}\right)}.
\end{align*}
Hence $z^2M^2-(1-z)M+1=0$ and
\begin{align*}
M(z)=\frac{1-z\pm \sqrt{1-2z-3z^2}}{2z^2}.
\end{align*}
Since $\epsilon \in L(\mathcal{M}_{\mathrm{Mot}})$, $M\rightarrow 0$ as $z\rightarrow 0$. Hence we take the negative sign for the radical in $M(z)$.
\end{example}

\begin{example}\label{ejeriordan}
The number of Motzkin paths of length $n$ without level steps on the $x$-axis (Riordan paths) is the $n$-th Riordan number $r_n$, sequence A005043. The number of words of length $n$  recognized by the convergent automaton  $\mathcal{M}_R$, see Figure \ref{autfinito}(left), is  the $n$-th Riordan number and its  GF is \begin{align}
R(z)&=\sum_{i=0}^{\infty}r_iz^{i}=\frac{1+z-\sqrt{1-2z-3z^2}}{2z(1+z)}\\
&=\cfrac{1}{1-\cfrac{z^2}{1-z-\cfrac{z^2}{1-z-\cfrac{z^2}{\ddots}}}}.
\end{align}

In the automaton  $\mathcal{M}_R$, the initial loop  was removed to avoid the level steps on the $x$-axis, then it is clear that $r_n=\left|\left\{w\in L(\mathcal{M}_{\mathrm{R}}): \left|w\right|=n \right\}\right|=L^{(n)}(\mathcal{M}_{\mathrm{R}})$. Moreover, we can use equivalent automata because the automaton  $\mathcal{M}_{\mathrm{Mot}}$ is a subautomaton of $\mathcal{M}_R$. Hence, there is an automaton  $\mathcal{M}'_R $ with only two visible states, such that $\mathcal{M}'_R \cong \mathcal{M}_R$, see Figure  \ref{autfinito}(right).

\begin{figure}[h]
  \centering
    \unitlength=2.8pt
    \begin{picture}(75, 9)(0,-1)
      \gasset{Nw=4,Nh=4,Nmr=2,curvedepth=0}
    \thinlines
    \node[Nmarks=ir,iangle=180, curvedepth=3](A0)(0,1){\tiny{$q_{0}$}}
    \node(A1)(15,1){\tiny{$q_{1}$}}
    \node(A2)(30,1){\tiny{$q_{2}$}}
    \node(A3)(45,1){\tiny{$q_{3}$}}
    \drawedge[curvedepth=2.5](A0,A1){$z$}
    \drawedge[curvedepth=2.5](A1,A0){$z$}
    \drawedge[curvedepth=2.5](A1,A2){$z$}
    \drawedge[curvedepth=2.5](A2,A1){$z$}
    \drawedge[curvedepth=2.5](A2,A3){$z$}
    \drawedge[curvedepth=2.5](A3,A2){$z$}
    \drawloop[loopdiam=4,loopangle=90](A1){$z$}
    \drawloop[loopdiam=4,loopangle=90](A2){$z$}
    \drawloop[loopdiam=4,loopangle=90](A3){$z$}
    \node[Nmarks=ir,iangle=180, curvedepth=3](A00)(65,1){\tiny{$q_{0}$}}
    \node(A11)(80,1){\tiny{$q_{1}$}}
    \drawedge[curvedepth=2.5](A00,A11){$z$}
    \drawedge[curvedepth=2.5](A11,A00){$zM(z)$}
    \gasset{Nframe=n,Nadjust=w,Nh=6,Nmr=0}
    \node(P)(51,1){$\cdots$}
    \node(P)(56,1){$\cong$}
      \node(P)(56,1){$\cong$}
    \        \node(P)(-6,8){$\mathcal{M}_R$:}
            \node(P)(60,8){$\mathcal{M}'_R$:}
        \end{picture}
  \caption{Equivalent Automata $\mathcal{M}'_R \cong \mathcal{M}_R$, Example \ref{ejeriordan}.} \label{autfinito}
\end{figure}
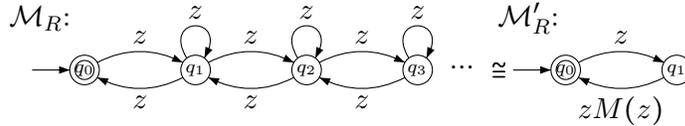

Then we have the following system of GFs equations
\begin{align*}
\left\{\begin{aligned}
R(z)&= 1+ zL_{1}\\
L_{1}&= zM(z)R(z),\end{aligned}
\right.
\end{align*}
where $M(z)$ is the GF for Motzkin numbers. Whence
\begin{align}\label{riordanec}
R(z)=1+z^2M(z)R(z),
\end{align}
then
\begin{align*}
R(z)=\frac{1}{1-z^2M(z)}=\frac{2}{1+z+\sqrt{1-2z-3z^2}}=\frac{1+z-\sqrt{1-2z-3z^2}}{2z(1+z)}
\end{align*}
Moreover, from Equation (\ref{riordanec})
\begin{align*}
r_n= \sum_{j=0}^{n-2}m_jr_{n-j-2}, \ \ n\geq 2.
\end{align*}
We also have  $R(z)=\frac{1+zM(z)}{1+z}$ then $(1+z)R(z)=1+zM(z)$. Hence $r_{n+1}+r_{n}=m_{n}, \ n\geq 0$, this equation is derived in  \cite{FRA} using a combinatorial argument.\end{example}

\section{Linear Infinite  Counting Automaton}
In this section we study a family of counting automata and their GFs.
\begin{definition}
A linear graph $G$ is a  4-tuple  $G=(V, A, n, F)$ where
\begin{itemize}
\item $V\subseteq \mathbb{Z}$ is a  nonempty set of the labelled vertices of $G$. If $m\in V$, the vertex is labelled  by $m$.
\item $A=(A_{-},  A_{\curvearrowright}, A_{\curvearrowleft})$  is the set of edges of $G$, where
\begin{itemize}
\item $A_{-}=\{i\in V: \text{there exists a loop at vertex $i$} \}$.
\item $A_{\curvearrowright}=\{i\in V: \text{there exists an edge from vertex $i$ to vertex $i+1$ } \}$.
\item $A_{\curvearrowleft}=\{i\in V: \text{there exists an edge from vertex $i+1$ to vertex $i$} \}$.
\end{itemize}
\item $n\in V$ is the initial vertex.
\item $F\subseteq V$ is the set of final vertices.
\end{itemize}
\end{definition}
In particular, if $V=\mathbb{N}$, $A=(\mathbb{N}, \mathbb{N}, \mathbb{N}), n=0, F=\{0\}$, we say that $G$ is a \emph{complete linear graph} and is denoted by  $G_{L}$. If $V=\mathbb{Z}$, $A=(\mathbb{Z}, \mathbb{Z}, \mathbb{Z}), n=0, F=\{0\}$, we say that $G$ is a \emph{complete bilinear graph} and is denoted by $G_{BT}$.

\begin{example}\label{glineje2}
Consider the  linear graph  $V=\mathbb{N}$, $A=(2\mathbb{N}, \mathbb{N}, \mathbb{N}), n=0, F=\{0\}$. It is displayed in Figure \ref{glinealeje2}.
\begin{figure}[h]
 \begin{center}
    \unitlength=3pt
    \begin{picture}(52, 6)(0,2)
        \gasset{Nw=4,Nh=4,Nmr=2,curvedepth=0}
    \thinlines
    \node[Nmarks=ir,iangle=180, curvedepth=3](A0)(0,1){\tiny{$0$}}
    \node(A1)(15,1){\tiny{$1$}}
    \node(A2)(30,1){\tiny{$2$}}
    \node(A3)(45,1){\tiny{$3$}}
    \drawedge[curvedepth=3](A0,A1){}
    \drawedge[curvedepth=3](A1,A0){}
    \drawedge[curvedepth=3](A1,A2){}
    \drawedge[curvedepth=3](A2,A1){}
    \drawedge[curvedepth=3](A2,A3){}
    \drawedge[curvedepth=3](A3,A2){}
    \gasset{Nframe=n,Nadjust=w,Nh=6,Nmr=0}
    \node(P)(52,1){$\cdots$}
       \drawloop[loopangle=90](A0){}
    \drawloop[loopangle=90](A2){}
    \end{picture}
  \end{center}
  \caption{Linear graph, Example \ref{glineje2}.}
  \label{glinealeje2}
\end{figure}
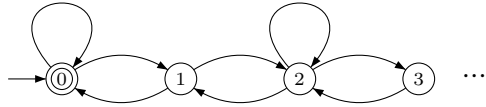
\end{example}

\begin{definition}
A linear counting automaton associated with the linear graph $G$ is a weighted automaton  $\mathcal{M}_c$ determined by the pair  $\mathcal{M}_c=(G, E)$, where $E$ is the set of weighted transitions defined by  the triple  $E=(E_{-},  E_{\curvearrowright}, E_{\curvearrowleft})$ where

\begin{itemize}
\item $E_{-}=\{h_i(z):  i \in A_{-}\}$.
\item $E_{\curvearrowright}=\{f_i(z): i\in A_{\curvearrowright} \}$.
\item $E_{\curvearrowleft}=\{g_i(z): i\in A_{\curvearrowleft} \}$.
\end{itemize}
and for all integer $i$,  $f_i(z), g_i(z)$ and $h_i(z)$ are transitions in parallel.
\end{definition}
The set of all counting automata is denoted by $\mathcal{M}^*_c$.

\begin{example} \label{ejeconteo}
Consider  the linear counting automaton $\mathcal{M}_c=(G, E=(\emptyset, \{z\}, \{z\}))$, where $G=(\mathbb{N}, (\emptyset, \mathbb{N}, \mathbb{N}), 0, \{0\})$. It is displayed in Figure \ref{conteolineal11}.
 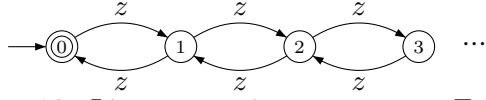
\begin{figure}[h]
 \begin{center}
    \unitlength=3pt
    \begin{picture}(52, 5)(0,1)
        \gasset{Nw=4,Nh=4,Nmr=2,curvedepth=0}
    \thinlines
    \node[Nmarks=ir,iangle=180, curvedepth=3](A0)(0,1){\tiny{$0$}}
    \node(A1)(15,1){\tiny{$1$}}
    \node(A2)(30,1){\tiny{$2$}}
    \node(A3)(45,1){\tiny{$3$}}
    \drawedge[curvedepth=3](A0,A1){$z$}
    \drawedge[curvedepth=3](A1,A0){$z$}
    \drawedge[curvedepth=3](A1,A2){$z$}
    \drawedge[curvedepth=3](A2,A1){$z$}
    \drawedge[curvedepth=3](A2,A3){$z$}
    \drawedge[curvedepth=3](A3,A2){$z$}
    \gasset{Nframe=n,Nadjust=w,Nh=6,Nmr=0}
    \node(P)(52,1){$\cdots$}
    \end{picture}
  \end{center}
  \caption{Linear counting automaton, Example \ref{ejeconteo}.}
  \label{conteolineal11}
\end{figure}
\end{example}

The linear infinite counting automaton associated with the complete linear graph  $G_{L}$ is denoted by $\mathcal{M}_{Lin}$,  see Figure \ref{conteolineal}(left). Similarly,  the linear infinite counting automaton associated with the complete bilinear graph  $G_{BL}$ is denoted by $\mathcal{M}_{BLin}$,  see Figure \ref{conteolineal}(right).
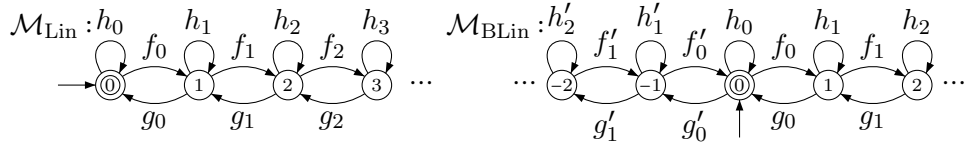
\begin{figure}[h]
  \centering
    \unitlength=2.8pt
    \begin{picture}(115, 15)(-7,-3)
        \gasset{Nw=4,Nh=4,Nmr=2,curvedepth=0}
    \thinlines
    \node[Nmarks=ir,iangle=180, curvedepth=3](A0)(0,0){\tiny{0}}
    \node(A1)(12,0){\tiny{1}}
    \node(A2)(24,0){\tiny{2}}
    \node(A3)(36,0){\tiny{3}}
    \drawedge[curvedepth=2.5](A0,A1){$f_0$}
    \drawedge[curvedepth=2.5](A1,A0){$g_0$}
    \drawedge[curvedepth=2.5](A1,A2){$f_1$}
    \drawedge[curvedepth=2.5](A2,A1){$g_1$}
    \drawedge[curvedepth=2.5](A2,A3){$f_2$}
    \drawedge[curvedepth=2.5](A3,A2){$g_2$}
    \drawloop[loopdiam=4,loopangle=90](A0){$h_0$}
    \drawloop[loopdiam=4,loopangle=90](A1){$h_1$}
    \drawloop[loopdiam=4,loopangle=90](A2){$h_2$}
    \drawloop[loopdiam=4,loopangle=90](A3){$h_3$}
     \node[Nmarks=ir,iangle=-90, curvedepth=3](A0)(85,0){\tiny{$0$}}
    \node(A1)(97,0){\tiny{$1$}}
    \node(A2)(109,0){\tiny{$2$}}
    \node(A11)(73,0){\tiny{$-1$}}
     \node(A22)(61,0){\tiny{$-2$}}
    \drawedge[curvedepth=2.5](A0,A1){$f_0$}
    \drawedge[curvedepth=2.5](A1,A0){$g_0$}
    \drawedge[curvedepth=2.5](A1,A2){$f_1$}
    \drawedge[curvedepth=2.5](A2,A1){$g_1$}
    \drawedge[curvedepth=2.5](A0,A11){$g'_0$}
    \drawedge[curvedepth=2.5](A11,A22){$g'_1$}
        \drawedge[curvedepth=2.5](A22,A11){$f'_1$}
    \drawedge[curvedepth=2.5](A11,A0){$f'_0$}
    \gasset{Nframe=n,Nadjust=w,Nh=6,Nmr=0}
    \node(P)(56,0){$\cdots$}
       \node(P)(114,0){$\cdots$}
    \node(P)(52,8){$\mathcal{M}_{\mathrm{BLin}}:$}
       \drawloop[loopdiam=4,loopangle=90](A0){$h_0$}
    \drawloop[loopdiam=4,loopangle=90](A1){$h_1$}
    \drawloop[loopdiam=4,loopangle=90](A2){$h_2$}
    \drawloop[loopdiam=4,loopangle=90](A11){$h'_1$}
       \drawloop[loopdiam=4,loopangle=90](A22){$h'_2$}
    \gasset{Nframe=n,Nadjust=w,Nh=6,Nmr=0}
    \node(P)(42,0){$\cdots$}
    \node(P)(-8,8){$\mathcal{M}_{\mathrm{Lin}}:$}
    \end{picture}
  \caption{Infinite counting automata $\mathcal{M}_{Lin}$ and $\mathcal{M}_{BLin}$.}
  \label{conteolineal}
\end{figure}

\subsection{Generating Function of  $\mathcal{M}_{Lin}$ and $\mathcal{M}_{BLin}$ }
\begin{theorem}\label{teoflajolet}
The generating function  of $\mathcal{M}_{\mathrm{Lin}}$, see Figure \ref{conteolineal} (left), is
\begin{align*}
E(z)&=\cfrac{1}{1-h_0\left(z\right)-\cfrac{f_0\left(z\right)g_0\left(z\right)}{1-h_1\left(z\right)- \cfrac{f_1\left(z\right)g_1\left(z\right)}{\ddots}}},
\end{align*}
where $f_i(z), g_i(z)$ and $h_i(z)$ are transitions in parallel for all integer $i\geqslant 0$.
\end{theorem}
\begin{proof}
We denoted by $\mathcal{M}_{Lin-s}$ the automaton obtained from $\mathcal{M}_{Lin}$ deleting the vertices $s, s+1, s+2, \dots$.  The system of GFs equations of $\mathcal{M}_{Lin-s}, s\geq1$, is
\begin{align*}
\left\{\begin{aligned}
L_{0}&= h_0L_{0}+ f_0L_{1} + 1\\
L_{i}&= g_{i-1}L_{i-1} + h_iL_{i} + f_iL_{i+1}, \ \  1\leq i \leq s-1 \\
L_{s}&= g_{s-1}L_{s-1} + h_sL_{s}.\end{aligned}
\right.
\end{align*}
Substituting repeatedly into each equation $L_i$, we have
\begin{align*}
L_0=\cfrac{1}{ 1-h_0-\cfrac{f_0g_0}{1-h_1-\cfrac{f_1g_1}{\cfrac{\vdots }{1-h_{s-1} - \cfrac{f_{s-1}g_{s-1}}{1-h_s}}}}}.
\end{align*}
Since $\mathcal{M}_{Lin}$ is convergent, then when $s\rightarrow \infty$ we obtain the convergent continued fraction $E(z)$ of the GF of  $\mathcal{M}_{Lin}$.
\end{proof}
The last theorem coincides with Theorem 1 in \cite{FLA} and Theorem 9.1 in \cite{RUT}. However, this presentation extends their applications, taking into account that $f_i(z), g_i(z)$ and $h_i(z)$ are GFs, which can be GFs of  several variables.

\begin{corollary}\label{coro1}
If for all integer $i\geq0$, $f_i(z)=f(z), g_i(z)=g(z)$ and $h_i(z)=h(z)$  in $\mathcal{M}_{\mathrm{Lin}}$, then the GF is
\begin{align}
B(z)&=\frac{1-h(z)-\sqrt{(1-h(z))^2-4f(z)g(z)}}{2f(z)g(z)} \label{ecmot6} \\
&=\sum_{n=0}^{\infty}\sum_{m=0}^{\infty}C_n \binom{m+2n}{m}\left(f\left(z\right)g\left(z\right)\right)^n\left(h(z)\right)^m \label{ecmot7}\\
&=\cfrac{1}{1-h\left(z\right)-\cfrac{f\left(z\right)g\left(z\right)}{1-h\left(z\right)- \cfrac{f\left(z\right)g\left(z\right)}{1-h\left(z\right) -\cfrac{f\left(z\right)g\left(z\right)}{\ddots}}}},
\end{align}
where $f(z), g(z)$ and $h(z)$ are transitions in parallel and $C_n$ is the $n$-th Catalan number, sequence A000108.
\end{corollary}
\begin{proof}
From Theorem \ref{teoflajolet} the GF is
\begin{align*}
B(z)=\cfrac{1}{1-h\left(z\right)-\cfrac{f\left(z\right)g\left(z\right)}{1-h\left(z\right)- \cfrac{f\left(z\right)g\left(z\right)}{1-h\left(z\right) -\cfrac{f\left(z\right)g\left(z\right)}{\ddots}}}}.
\end{align*}
Since the counting automaton is convergent, then
 \begin{align*}
B(z)=\frac{1}{1-h(z)-f(z)g(z)B(z)}.
\end{align*}
Hence
\begin{align*}
B(z)=\frac{1-h(z)-\sqrt{(1-h(z))^2-4f(z)g(z)}}{2f(z)g(z)}.
\end{align*}
Equation (\ref{ecmot7})  is obtained from observing that
\begin{align*}
B(z)&=\frac{1-h(z)-\sqrt{(1-h(z))^2-4f(z)g(z)}}{2f(z)g(z)}\\
&=\frac{1}{1-h(z)}\frac{1-\sqrt{1-4\frac{f(z)g(z)}{(1-h(z))^2}}}{2\frac{f(z)g(z)}{(1-h(z))^2}}\\
&=\frac{1}{1-h(z)}C(u),
\end{align*}
where $u=\frac{f(z)g(z)}{(1-h(z))^2}$ and $C(z)=\frac{1-\sqrt{1-4z}}{2z}$ is the GF for Catalan numbers. Therefore
\begin{align*}
B(z)=\frac{1}{1-h(z)}C(u)&=\frac{1}{1-h(z)}\sum_{n=0}^{\infty}C_nu^n\\
&=\sum_{n=0}^{\infty}C_n\frac{\left(f(z)g(z)\right)^n}{(1-h(z))^{2n+1}}\\
&=\sum_{n=0}^{\infty}\sum_{m=0}^{\infty}C_n\binom{m+2n}{m}\left(f(z)g(z)\right)^n\left(h(z)\right)^m.
\end{align*}
\end{proof}

\begin{example}  If $h(z)=f(z)=g(z)=z$ in Corollary  \ref{coro1},  we obtain the GF for Motzkin paths  $M(z)$, see Example  \ref{ejmotzkin}. Moreover,
  \begin{align*}
  m_s=\sum_{n=0}^{\left\lfloor\frac{s}{2}\right\rfloor}C_n\binom{s}{2n}.
  \end{align*}
Indeed, using the Equation (\ref{ecmot7})  it follows that
\begin{align*}
M(z)=\sum_{n=0}^{\infty}\sum_{m=0}^{\infty}C_n\binom{m+2n}{m}z^{2n+m},
\end{align*}
taking $t=2n+m$
\begin{align*}
M(z)=\sum_{n=0}^{\infty}\sum_{t=2n}^{\infty}C_n\binom{t}{t-2n}z^{t},
\end{align*}
hence
  \begin{align*}
  m_s=\sum_{n=0}^{\left\lfloor\frac{s}{2}\right\rfloor}C_n\binom{s}{s-2n}=\sum_{n=0}^{\left\lfloor\frac{s}{2}\right\rfloor}C_n\binom{s}{2n}.
  \end{align*}
  \end{example}

\subsection{Applications of CAM to  $k$-colored Motzkin paths and generalized Motzkin paths}

\begin{example}
A $k$-colored Motzkin path of length $n$ is a Motzkin path such that the level steps are labelled by $k$ colors. The number of $k$-colored Motzin paths of length $n$  is the $n$-th $k$-colored Motzkin number, $m_{n, k}$.

If $f(z)=z=g(z)$  and $h(z)=kz$ in Corollary \ref{coro1},  we obtain the GF for $k$-colored Motzkin paths
\begin{align*}
M_k(z)=\sum_{i=0}^{\infty}m_{i, k}z^{i}=\frac{1-kz-\sqrt{(1-kz)^2-4z^2}}{2z^2},
\end{align*}
and
\begin{align}\label{ecmotzkincol}
m_{n,k}=\sum_{n=0}^{\left\lfloor\frac{s}{2}\right\rfloor}C_n\binom{s}{2n}k^{s-2n}.
\end{align}
Equation (\ref{ecmotzkincol}) coincides with Equation (8) of  \cite{SAP}.
In particular, if $k=2$
\begin{align*}
M_2(z)&=\sum_{i=0}^{\infty}m_{i, 2}z^{i}=\frac{1-2z-\sqrt{1-4z}}{2z}\\
&=z+2z^2+5z^3+14z^4+42z^5+132z^6+429z^7+1430z^8+\cdots
\end{align*}
then $zM_2(z)=C(z)-1$ so the $n$-th 2-colored Motzkin number is equal to  $C_{n+1}$.  In \cite{Deumot} and \cite{WEN} there are some bijections between the numbers   $m_{n,2}$ and other combinatorial objects.  If $k=3$
\begin{align*}
M_3(z)&=\sum_{i=0}^{\infty}m_{i, 3}z^{i}=\frac{1-3z-\sqrt{1-6z+5z^2}}{2z}\\
&=z+3z^2+10z^3+36z^4+137z^5+543z^6+2219z^7+9285z^8+\cdots,
\end{align*}
then the sequence A002212 is obtained, some properties about these numbers are established in  \cite{DEU, SHAC}.  If $k=4$ the sequence A005572 which starts 1, 4, 17, 76, 354, 1704, 8421, 42508  is obtained.
\end{example}

\begin{example}\label{generalmot}
A generalized Motzkin path of length $n$ is a Motzkin path  such that the level step is  $H=(k, 1)$ where  $k$ is a fixed positive integer. The number of generalized Motzkin paths of length $n$ is the $n$-th generalized Motzkin number $m_{n}^{k}$.
If $f(z)=z=g(z)$ and $h(z)=z^k$ in Corollary \ref{coro1},  we obtain the GF  for generalized  Motzkin path.
  \begin{align*}
M^{(k)}(z)=\sum_{i=0}^{\infty}m_{i}^{k}z^{i}=\frac{1-z^k-\sqrt{(1-z^k)^2-4z^2}}{2z^2}.
\end{align*}
The last equation coincides with the Equation (1) of  \cite{ECO1}. If $k=2$, we obtain the GF for Schröder paths, sequence A006318.
 \begin{align*}
M^{(2)}(z)&=\sum_{i=0}^{\infty}m_{i}^{2}z^{i}=\frac{1-z^2-\sqrt{1 - 6z^2 + z^4}}{2z^2}\\
&=1 + 2z^2 + 6z^4 + 22z^6 + 90z^8 + 394z^{10} + 1806z^{12} + \cdots
\end{align*}
There is a relation well known between the numbers $m_{i}^2$ and Narayana numbers $N(n,k)=\frac{1}{n}\binom{n}{k}\binom{n}{k-1}$ with $1\leqslant k \leqslant n$  which enumerate a large variety of combinatorial objects, see sequence A001263. In particular,  there is the following identity, \cite{WILL}
\begin{align*}
m_n^2=\sum_{k=0}^{n}N(n,k)2^k.
\end{align*}

 \end{example}

\begin{example}\label{varmot}
If $f(z)=z=g(z)$ and $h(z)=\frac{kz}{1-z}$ in Corollary \ref{coro1}, we obtain the GF $F_k(z)$ for the lattice path which never goes below the $x$-axis, from $(0,0)$ to $(n,0)$ consisting of up steps $U = (1,1)$, down steps $D = (1,-1)$ and horizontal steps $H(k) = (k,0)$ for every positive integer $k$ and can be labelled by $k$ colors
\begin{align*}
F_k(z)=\frac{1-(1+k)z-\sqrt{1-(2+2k)z+(-3+2k+k^2)z^2+8z^3-4z^4}}{2z^2(1-z)},
\end{align*}
and
\begin{align*}
f^{(k)}_s=\sum_{n=0}^{s}\sum_{m=0}^{s-2n}C_n\binom{m+2n}{m}\binom{s-2n-1}{s-m-2n}k^m,
\end{align*}
where $f^{(k)}_s=\left[z^s\right]F_k(z)$. From Corollary  \ref{coro1} we have the first equation, besides
\begin{align*}
F_k(z)&=\sum_{n=0}^{\infty}\sum_{m=0}^{\infty}C_n\binom{m+2n}{m}z^{2n}\left(\frac{kz}{1-z}\right)^m\\
&=\sum_{n=0}^{\infty}\sum_{m=0}^{\infty}\sum_{i=0}^{\infty}C_n\binom{m+2n}{m}\binom{i+m-1}{i}k^mz^{2n+m+i},
\end{align*}
taking $t=2n+m+i$
\begin{align*}
\sum_{n=0}^{\infty}\sum_{m=0}^{\infty}\sum_{t=2n+m}^{\infty}C_n\binom{m+2n}{m}\binom{t-2n-1}{t-m-2n}k^mz^{t},
\end{align*}
hence
\begin{align*}
f^{(k)}_s=\sum_{n=0}^{s}\sum_{m=0}^{s-2n}C_n\binom{m+2n}{m}\binom{s-2n-1}{s-m-2n}k^m.
\end{align*}
If $k=1$ the sequence A135052 is obtained.
\end{example}
The last example shows the variety of results that can be obtained with the CAM, simply by changing  the different transitions in parallel, i.e., changing  the GFs $f_i(z), g_i(z) $ and $h_i(z)$.

\begin{definition}\label{fgoei} For all integer $i\geq 0$ we define the continued fraction
 $E_i(z)$ by:
\begin{align*}
E_i(z)&=\cfrac{1}{1-h_i\left(z\right)-\cfrac{f_{i}\left(z\right)g_{i}\left(z\right)}{1-h_{i+1}\left(z\right)- \cfrac{f_{i+1}\left(z\right)g_{i+1}\left(z\right)}{1-h_{i+2}\left(z\right) -\cfrac{f_{i+2}\left(z\right)g_{i+2}\left(z\right)}{\ddots}}}},
\end{align*}
where $f_i(z), g_i(z), h_i(z)$ are transitions in parallel for all  integers positive $i$.
\end{definition}

\begin{theorem}\label{teoflajoletbi}
The  generating function  of $\mathcal{M}_{\mathrm{BLin}}$, see Figure \ref{conteolineal}(right), is
\begin{align*}\label{fgobilineal}
E_b(z)&=\cfrac{1}{1-h_0(z)-f_0(z)g_0(z)E_1(z)-f'_0(z)g'_0(z)E'_1(z)},
\end{align*}
where $f_i(z), f'_i(z), g_i(z), g'_i(z), h_i(z)$ and $h'_i(z)$ are transitions in parallel for all $i\in\mathbb{Z}$.
\end{theorem}

\begin{proof}
It is clear that the automaton $\mathcal{M}_{\mathrm{BLin}}$ is equivalent to the automaton in  Figure \ref{conteobilinealiso}.
 \begin{figure}[h]
 \begin{center}
    \unitlength=3pt
    \begin{picture}(40, 12)(-20,-1)
        \gasset{Nw=4,Nh=4,Nmr=2,curvedepth=0}
    \thinlines
    \node[Nmarks=ir,iangle=90, curvedepth=3](A0)(0,1){\tiny{$0$}}
    \node(A1)(20,1){\tiny{$1$}}
    \node(A11)(-20,1){\tiny{$-1$}}
    \drawedge[curvedepth=3](A0,A1){$f_0$}
    \drawedge[curvedepth=3](A1,A0){$g_0E_1$}
    \drawedge[curvedepth=3](A0,A11){$g'_0$}
    \drawedge[curvedepth=3](A11,A0){$f'_0E'_1$}
    \gasset{Nframe=n,Nadjust=w,Nh=6,Nmr=0}
       \drawloop[loopangle=90](A0){$h_0$}
    \end{picture}
  \end{center}
  \caption{Equivalent automaton to $\mathcal{M}_{BLin}$.}
  \label{conteobilinealiso}
\end{figure}
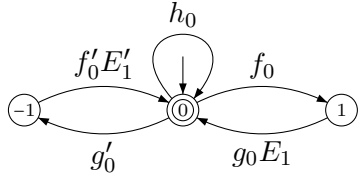

Therefore, we have the following system of GFs equations
\begin{align*}
\left\{\begin{aligned}
L_{0}&= h_0L_{0}+ f_0L_{1} + g'_0L'_1 + 1 \\
L_{1}&= g_{0}E_1L_{0} \\
L'_{1}&= f_{0}E'_{1}L_{0},\end{aligned}
\right.
\end{align*}
where $L'_1=L_{-1}$. Solving the system for $L_0$ we obtain the GF of $\mathcal{M}_{\mathrm{BLin}}$.
\end{proof}

\begin{corollary}\label{corobi}
If for all integer  $i$, $f_i(z)=f(z)=f'_i(z), g_i(z)=g(z)=g'_i(z)$ and $h_i(z)=h(z)=h'_i(z)$  in $\mathcal{M}_{\mathrm{BLin}}$, then we have the GF
\begin{align}
B_b(z)&=\frac{1}{\sqrt{(1-h(z))^2-4f(z)g(z)}}\\
&=\cfrac{1}{1-h(z) - \cfrac{2f(z)g(z)}{1-h(z)-\cfrac{f(z)g(z)}{1-h(z)-\cfrac{f(z)g(z)}{\ddots}}}},
\end{align}
where $f(z), g(z)$ and $h(z)$ are transitions in parallel.
\end{corollary}
\begin{proof}
It is clear from Equation (\ref{fgobilineal})
\begin{align*}
B_b(z)&=\frac{1}{1-h(z)-2f(z)g(z)B(z)},
\end{align*}
where $B(z)$ is the GF in Corollary \ref{coro1}. The continued fraction is obtained from the same corollary.
\end{proof}
\begin{corollary}\label{corobi2}
If  $f(z)=g(z)$  in Corollary \ref{corobi}, then we have the GF
\begin{align}
B_b(z)&=\frac{1}{1-h(z)} + \sum_{n=1}^{\infty}\sum_{k=0}^{\infty}\sum_{l=0}^{\infty}2^n\frac{n}{n+2k}\binom{n+2k}{k}\binom{l+2n+2k}{l}f(z)^{2n+2k}h(z)^{l},
\end{align}
where $f(z), g(z)$ and $h(z)$ are  transitions in parallel.
\end{corollary}
\begin{proof}
From Theorem  \ref{teoflajoletbi} we have
\begin{align*}
B_b(z)=\frac{1}{1-h(z)-2f^2(z)B(z)},
\end{align*}
where $B(z)$ is the GF in Corollary  \ref{coro1} with $g(z)=f(z)$, i.e.,
\begin{align*}
B(z)&=\frac{1-h(z)-\sqrt{(1-h(z))^2-4f^2(z)}}{2f^2(z)}=\frac{1}{1-h(z)}C(u),
\end{align*}
where $u=\frac{f^2(z)}{(1-h(z))^2}$ and $C(u)$ is the GF for Catalan numbers.  The powers of the GF for Catalan numbers (see Eq. 5.70 of \cite{knuth}), satisfy that
\begin{align*}
C^n(u)=\sum_{k=0}^{\infty}\frac{n}{n+2k}\binom{n+2k}{k}u^k, \ \ n\geqslant1.
\end{align*}
Then
\begin{align*}
B_b(z)&=\frac{1}{1-h(z)-2f^2(z)\frac{C(u)}{1-h(z)}}=\frac{1}{1-h(z)}\frac{1}{1-2uC(u)}\\
&=\frac{1}{1-h(z)}\sum_{n=0}^{\infty}(2uC(u))^n=\frac{1}{1-h(z)} + \sum_{n=1}^{\infty}(2uC(u))^n\\
&=\frac{1}{1-h(z)} + \sum_{n=1}^{\infty}\sum_{k=0}^{\infty}2^n\frac{n}{n+2k}\binom{n+2k}{k}\frac{f^{2n}(z)}{(1-h(z))^{2n}} u^k\\
&=\frac{1}{1-h(z)} + \sum_{n=1}^{\infty}\sum_{k=0}^{\infty}2^n\frac{n}{n+2k}\binom{n+2k}{k}\frac{f^{2n+2k}(z)}{(1-h(z))^{2n+2k}}\\
&=\frac{1}{1-h(z)} + \sum_{n=1}^{\infty}\sum_{k=0}^{\infty}\sum_{l=0}^{\infty}2^n\frac{n}{n+2k}\binom{n+2k}{k}\binom{l+2n+2k}{l}f^{2n+2k}(z)h^{l}(z).
\end{align*}
\end{proof}

\begin{example}\label{ejmotzkinbi}
A  Grand Motzkin path of length  $n$ is a Motzkin path without the condition that never passes below the $x$-axis.  The number of Grand Motzkin paths of length $n$ is the $n$-th Grand Motzkin number $m^b_{n}$, sequence A002426.

The number of words of length $n$ recognized by the convergent automaton $\mathcal{M}_{\mathrm{BLin}}$, see Figure \ref{conteolineal}(right), with $f_i(z)=z=g_i(z)=h_i(z)=f'_i(z)=g'_i(z)=h'_i(z)$ for all integer $i$, is the $n$-th Grand Motzkin number and its GF is
\begin{align}
M^b\left(z\right)&=\sum_{i=0}^{\infty}m^b_{i}z^{i}\frac{1}{\sqrt{1-2z-3z^2}}\\
&=\cfrac{1}{1-z-\cfrac{2z^2}{1-z-\cfrac{z^2}{1-z-\cfrac{z^2}{\ddots}}}}\\
&=\frac{1}{1-z} + \sum_{n=1}^{\infty}\sum_{k=0}^{\infty}\sum_{l=0}^{\infty}2^n\frac{n}{n+2k}\binom{n+2k}{k}\binom{l+2n+2k}{l}z^{2n+2k+l}.  \label{ecbimot}
\end{align}
These equations are easily obtained from Corollary  \ref{corobi} and \ref{corobi2}.

In this case the edge from the state $i$ and $i+1$ and vice versa, represent a rise or a fall  above the $x$-axis and the edge from the state $-i$ and $-(i+1)$ and vice versa, represent  a rise or a fall below the $x$-axis, and the loops represent the level steps.   Moreover, it is clear that a word  is recognized by $\mathcal{M}_{\mathrm{BLin}}$ if and only if it has an equal number of steps to the right and to the left, then
\begin{align*}
m^b_n=\left|\left\{w\in L(\mathcal{M}_{\mathrm{BLin}}): \left|w\right|=n \right\}\right|=L^{(n)}(\mathcal{M}_{\mathrm{BLin}}).
\end{align*}
On the other hand, taking $t=2n+2k+l$  in  Equation (\ref{ecbimot}), we have
 \begin{align*}
\sum_{i=0}^{\infty}m^b_{i}z^{i}&=\frac{1}{1-z} + \sum_{n=1}^{\infty}\sum_{k=0}^{\infty}\sum_{t=2n+2k}^{\infty}2^n\frac{n}{n+2k}\binom{n+2k}{k}\binom{t}{t-2n-2k}z^{t},
\end{align*}
then
\begin{align*}
m^b_s=1+\sum_{n=0}^{s}\sum_{k=0}^{\lfloor \frac{s-2n}{2} \rfloor}2^n\frac{n}{n+2k}\binom{n+2k}{k}\binom{s}{2n+2k}.
\end{align*}

The Grand Motzkin  paths are related to the central trinomial coefficients. Let $T_n$ denote
the $n$-th central trinomial coefficient, defined as the coefficient of $x^n$ in the expression of $(1 + x + x^2)^n$  or it can also  be defined as the coefficient of the form $x^ny^nz^k$ in the expression of $(x+y+z)^n$. For example if $n=2$ then   $(x+y+z)^2=x^2+y^2+z^2+2xy+2xz+2yz$, hence $T_2=3$.  It is clear that the Grand Dyck paths are enumerated by the central binomial coefficients, i.e.,  $T_n=m_n^b$.  For integers $a, b, c$ we call the coefficient of $x^n$ in the expression  $(a+bx+cx^2)^n$ the generalized central trinomial coefficient, $T^*_n$ or it can also be defined as the coefficients of the form $x^ny^nz^k$ in the expression $(a+bx+cx^2)^n$.  Taking $f(z)=az=f'(z), g(z)=cz=g'(z)$ and $h(z)=bz=h'(z)$ in the counting automaton  $\mathcal{M}_{\mathrm{BLin}}$, we obtain the GF for the numbers  $T^*_n$.
 \begin{align*}
 \sum_{i=0}^{\infty} T^*_iz^i&=\frac{1}{\sqrt{(1-bz)^2-4acz^2}}=\frac{1}{\sqrt{1-2bz+(b^2-4ac)z^2}}\\
 &=\cfrac{1}{1-bz-\cfrac{2acz^2}{1-bz-\cfrac{acz^2}{1-bz-\cfrac{acz^2}{\ddots}}}}
 \end{align*}
The last GF coincides with Equation (3) of \cite{NOE}.  By using the Binomial Theorem twice and the identity $\binom{n}{k}\binom{n-k}{n-2k}=\binom{2k}{k}\binom{n}{2k}$, we have
  \begin{align*}
T^*_n&=\sum_{k=0}^{\lfloor n/2\rfloor}\binom{2k}{k}\binom{n}{2k}b^{n-2k}(ac)^k.
 \end{align*}
Since $m_n^b=T_n$, then
 \begin{align*}
  \sum_{k=0}^{\lfloor s/2\rfloor}\binom{2k}{k}\binom{s}{2k}=1+\sum_{n=0}^{s}\sum_{k=0}^{\lfloor \frac{s-2n}{2} \rfloor}2^n\frac{n}{n+2k}\binom{n+2k}{k}\binom{s}{2n+2k}.
 \end{align*}
\end{example}

\end{document}